\documentclass[11pt,a4paper]{article}
\usepackage{amsmath, amsfonts, amssymb}
\usepackage{a4wide}
\usepackage{parskip}
\usepackage{enumitem}
\usepackage{xcolor}

\usepackage{hyperref}
\usepackage{booktabs}
\usepackage{caption}
\usepackage{siunitx}
\usepackage{tabularx}
\usepackage{comment}
\usepackage{multirow}
\usepackage{adjustbox}
\usepackage{mathtools}
\usepackage{caption}
\usepackage{parskip}
\usepackage{graphicx}
\usepackage{adjustbox}
\usepackage{multirow}
\usepackage{amsthm}
\usepackage{bm, makecell}
\usepackage{lscape}
\usepackage{rotating}
\usepackage{floatrow}
\usepackage[noadjust]{cite}

\newtheorem{theorem}{Theorem}[section]
\newtheorem{lemma}[theorem]{Lemma}
\newtheorem{corollary}[theorem]{Corollary}
\newtheorem{proposition}[theorem]{Proposition}
\theoremstyle{definition}
\newtheorem{definition}[theorem]{Definition}

\newtheorem{remark}[theorem]{Remark}

\newcommand{\Ann}{\operatorname{Ann}}

\usepackage{tikz,xcolor,hyperref}
\usepackage{mathdots}
\definecolor{lime}{HTML}{A6CE39}
\DeclareRobustCommand{\orcidicon}{%
	\begin{tikzpicture}
		\draw[lime, fill=lime] (0,0) 
		circle [radius=0.16] 
		node[white] {{\fontfamily{qag}\selectfont \tiny ID}};
		\draw[white, fill=white] (-0.0625,0.095) 
		circle [radius=0.007];
	\end{tikzpicture}
	\hspace{-2mm}
}

\foreach \x in {A, ..., Z}{%
	\expandafter\xdef\csname orcid\x\endcsname{\noexpand\href{https://orcid.org/\csname orcidauthor\x\endcsname}{\noexpand\orcidicon}}
}

\begin{document}
\date{}
\title{Hermitian Duality and Syndrome Structure of CSS-Type Codes over the ring $\mathbb{Z}_q[i]$}
		\author{{\bf Akanksha Tiwari  \footnote{email: {\tt akankshafzd8@gmail.com} }\orcidA{} and \bf Ritumoni Sarma\footnote{	email: {\tt ritumoni407@gmail.com}}\orcidC{}} \\\\ $*,\dagger$ Department of Mathematics \\ Indian Institute of Technology Delhi\\Hauz Khas, New Delhi-110016, India\\\\ }
\maketitle

\begin{abstract}

We study the algebraic structure of principal ideal codes and
CSS-type stabilizer constructions over the ring
$
R_q=\mathbb Z_q[i],$ where every prime divisor of $q$ is
congruent to $1$ modulo $4$. Although such ring-based constructions
have recently been considered in the context of classical and
quantum error correction, the precise relationship among principal
ideal cardinalities, Hermitian duality, algebraic cosets, and
stabilizer syndromes requires further structural analysis. 

We obtain
cardinality formulas for principal ideals and their annihilators.
For a length-one principal ideal code $C=\langle\alpha\rangle$, we
show that its Hermitian dual is
$
C^{\perp_H}=\operatorname{Ann}(\sigma(\alpha)).
$ For nested codes satisfying
$
C_2^{\perp_H}\subseteq C_1\subseteq C_2,
$
we determine the corresponding CSS-type stabilizer and prove that
the kernel of the physical $X$-error syndrome map is
$
\ker(\operatorname{Syn}_X)=C_2.
$
Consequently, $C_2/C_1$ parametrizes logical $X$-operator classes and $R_q^n/C_2$ parametrizes physical $X$-error syndrome
classes. Using this syndrome quotient, we construct a transversal of $C_2$ in
$R_q^n$ whose elements have pairwise distinct physical $X$-error
syndromes. The corresponding $X$-type error family is correctable by
stabilizer syndrome measurement, yielding a syndrome-based recovery
procedure that is consistent with the stabilizer structure. Explicit examples illustrate the resulting duality,
cardinality, and syndrome structure.

\end{abstract}

\noindent\textbf{Keywords:}
Quantum Error Correction;
Principal Ideal Codes; CSS codes; Stabilizer Codes; Syndrome
Structure; Logical Operators.

\medskip

\noindent\textbf{2020 Mathematics Subject Classification:}
94B05, 94B15, 94B35, 94B60.

\section{Introduction}

Quantum error correction provides an algebraic framework for
protecting quantum information against errors arising during
storage, transmission, and computation. Among the fundamental
constructions are stabilizer codes~\cite{gottesman1997stabilizer}
and Calderbank--Shor--Steane (CSS) codes
\cite{calderbank1996good,steane1996multiple}, originally formulated
over binary fields and subsequently extended to nonbinary finite
fields~\cite{ashikhmin2001nonbinary}. The connection between coding
theory over finite rings and nonlinear binary codes was highlighted
by Hammons et al.~\cite{hammons1994z}, motivating further study of
ring-based coding structures.

Quantum codes over rings have received comparatively limited
attention. Hostens et al.~\cite{PhysRevA.71.042315} developed a
stabilizer formalism over $\mathbb Z_d$, while Güzeltepe
\cite{Ozen2010QuantumCF} studied quantum codes using the Mannheim
metric. More recently, Güzeltepe and Aytaç
\cite{GUZELTEPE2027102881} proposed a CSS-type construction over the ring $\mathbb{Z}_q[i]$ together with a coset-based decoding
framework. The present work revisits the underlying algebraic and
stabilizer-theoretic structure, with particular attention to
principal ideal cardinality, Hermitian duality, and the
interpretation of physical syndrome distinguishability.

The present work revisits the algebraic and stabilizer-theoretic
foundations of CSS-type constructions over the ring $
R_q=\mathbb{Z}_q[i].$
Under the assumption that every prime divisor of $q$ is congruent to
$1$ modulo $4$, we use the Chinese remainder decomposition of $R_q$
to obtain a valuation-based description of principal ideals and their
cardinalities. This framework also leads to an annihilator-based
description of the Hermitian dual of a principal ideal code.

For nested codes
$C_2^{\perp_H}\subseteq C_1\subseteq C_2\subseteq R_q^n,$
we examine the associated CSS-type stabilizer structure and determine
the physical $X$-error syndrome space. A central point is that the
quotient $C_2/C_1$ and the physical syndrome quotient play different
roles: the former describes logical $X$-operator classes modulo
$X$-type stabilizers, whereas the latter is determined by the kernel
of the stabilizer syndrome map. This distinction is essential when
interpreting coset representatives as physical error representatives.

Using the resulting syndrome structure, we further construct a
transversal of the physical syndrome classes and obtain a correctable
family of $X$-type errors together with a corresponding syndrome-based
recovery procedure. Explicit examples over rings $\mathbb{Z}_q[i]$
illustrate the distinction between logical and syndrome quotients and
the applicability of the CRT--valuation cardinality formula.

The remainder of the paper is organized as follows.
Section~2 describes the structure of the ring
$R_q$.  Section~3 develops principal ideal cardinalities,
annihilators and Hermitian duality.  Section~4 studies
non-degenerate linear functionals and faithful characters.
Section~5 discusses nested codes and their quotient structure.
Section~6 develops the CSS-type stabilizer construction, determines its
syndrome and logical-operator structure, establishes the stabilizer
correctability criterion, and constructs a syndrome-based $X$-error
recovery framework using representatives of $R_q^n/C_2$.
Section~7 provides explicit examples, including a revisit of the coset
construction over $\mathbb{Z}_{85}[i]$. Section~8 concludes the paper.
\section{Structure of the Ring \texorpdfstring{$\mathbb Z_q[i]$}{Zq[i]}}
\label{sec:ring-structure}

Throughout the paper, let
$
R_q
=
\mathbb Z_q[i]
=
\mathbb Z_q[x]/\langle x^2+1\rangle.
$
Equivalently,
$
R_q
=
\{a+bi:\ a,b\in\mathbb Z_q,\ i^2=-1\}.
$
The ring $R_q$ is a finite commutative ring with identity and
$
|R_q|=q^2.
$

Although $R_q$ is defined for every positive integer $q$, throughout
the main results of this paper we assume that every prime divisor of
$q$ is congruent to $1$ modulo $4$. This hypothesis will first be
used in Section~2.2 to obtain the splitting decomposition.

Further, define
$
\sigma:R_q\longrightarrow R_q
$
by
$
\sigma(a+bi)=a-bi.
$
The map $\sigma$ is an involutive ring automorphism of $R_q$.
\subsection{Prime-power decomposition}
Write $
q=\prod_{j=1}^{s}p_j^{e_j},
$ where the $p_j$ are distinct primes and $e_j\geq1$.
By the Chinese Remainder Theorem,
$
\mathbb Z_q
\cong
\prod_{j=1}^{s}\mathbb Z_{p_j^{e_j}}.
$
\begin{theorem}[Prime-power decomposition]
\label{thm:prime-power-decomposition}
There is a ring isomorphism
\[
R_q
\cong
\prod_{j=1}^{s}\mathbb Z_{p_j^{e_j}}[i].
\]
\end{theorem}

\begin{proof}
Using
\[
R_q
=
\mathbb Z_q[x]/\langle x^2+1\rangle
\]
and the Chinese remainder decomposition of $\mathbb Z_q$, we
obtain
\[
\mathbb Z_q[x]
\cong
\prod_{j=1}^{s}\mathbb Z_{p_j^{e_j}}[x].
\]
Passing to the quotient by $x^2+1$ componentwise gives
\[
\mathbb Z_q[x]/\langle x^2+1\rangle
\cong
\prod_{j=1}^{s}
\mathbb Z_{p_j^{e_j}}[x]/\langle x^2+1\rangle,
\]
which is precisely
\[
R_q
\cong
\prod_{j=1}^{s}\mathbb Z_{p_j^{e_j}}[i].
\]
\end{proof}

\subsection{Splitting when \texorpdfstring{$p_j\equiv1\pmod4$}{pj ≡ 1 (mod 4)}}

From this point onward, unless stated otherwise, we assume that
every prime divisor $p_j$ of $q$ satisfies
\[
p_j\equiv1\pmod4.
\]

Since $p_j\equiv 1\pmod 4$, the polynomial $x^2+1$ has a root
modulo $p_j$. Since $p_j$ is odd, this root lifts to a root modulo
$p_j^{e_j}$. Choose $\omega_j\in\mathbb Z_{p_j^{e_j}}$ such that
$
\omega_j^2=-1.
$
\begin{theorem}[Splitting decomposition]
\label{thm:splitting-decomposition}
For every $j$,
$
\mathbb Z_{p_j^{e_j}}[i]
\cong
\mathbb Z_{p_j^{e_j}}
\times
\mathbb Z_{p_j^{e_j}}
$
via
\[
a+bi
\longmapsto
\bigl(a+b\omega_j,\ a-b\omega_j\bigr).
\]
Consequently,
\[
R_q
\cong
\prod_{j=1}^{s}
\left(
\mathbb Z_{p_j^{e_j}}
\times
\mathbb Z_{p_j^{e_j}}
\right).
\]
\end{theorem}

\begin{proof}
Since
$
x^2+1
=
(x-\omega_j)(x+\omega_j)
$
in $\mathbb Z_{p_j^{e_j}}[x]$, it is enough to verify that the two
factors are comaximal. Their difference is
$
(x+\omega_j)-(x-\omega_j)=2\omega_j.
$
Because $p_j$ is odd, $2$ is a unit in
$\mathbb Z_{p_j^{e_j}}$, and $\omega_j$ is also a unit since $\omega_j^2=-1.$
Hence $2\omega_j$ is a unit, so the ideals $
\langle x-\omega_j\rangle\text{ and }\langle x+\omega_j\rangle
$ are comaximal.\\
The Chinese Remainder Theorem therefore gives
\[
\mathbb Z_{p_j^{e_j}}[x]/\langle x^2+1\rangle
\cong
\mathbb Z_{p_j^{e_j}}[x]/\langle x-\omega_j\rangle
\times
\mathbb Z_{p_j^{e_j}}[x]/\langle x+\omega_j\rangle.
\]
Each factor is naturally isomorphic to
$\mathbb Z_{p_j^{e_j}}$, and the resulting evaluation map is
\[
a+bi
\longmapsto
\bigl(a+b\omega_j,\ a-b\omega_j\bigr).
\]
Combining this with
Theorem~\ref{thm:prime-power-decomposition} proves the result.
\end{proof}
\noindent
For later use, denote the CRT isomorphism by
$
\Phi:R_q
\longrightarrow
\prod_{j=1}^{s}
\left(
\mathbb Z_{p_j^{e_j}}
\times
\mathbb Z_{p_j^{e_j}}
\right).
$
For
$
\alpha=a+bi\in R_q,
$
write
$
\Phi(\alpha)
=
\bigl(
\alpha_{1,+},\alpha_{1,-},
\ldots,
\alpha_{s,+},\alpha_{s,-}
\bigr),
$
where
$
\alpha_{j,+}
=
a+b\omega_j\text{ and }
\,
\alpha_{j,-}
=
a-b\omega_j.
$
\subsection{Conjugation under the CRT decomposition}

The splitting decomposition also gives a particularly useful
description of conjugation.

\begin{proposition}
\label{prop:conjugation-crt}
Let
$
\Phi(\alpha)
=
\bigl(
\alpha_{1,+},\alpha_{1,-},
\ldots,
\alpha_{s,+},\alpha_{s,-}
\bigr).
$
Then
$
\Phi(\sigma(\alpha))
=
\bigl(
\alpha_{1,-},\alpha_{1,+},
\ldots,\\
\alpha_{s,-},\alpha_{s,+}
\bigr).
$
Thus conjugation interchanges the two components associated with
each prime-power factor.
\end{proposition}

\begin{proof}
Let
$
\alpha=a+bi.
$
Then
$
\sigma(\alpha)=a-bi.
$
Hence, in the $j$th pair,
$
\Phi_j(\sigma(\alpha))
=
\bigl(
a-b\omega_j,\ a+b\omega_j
\bigr)
=
\bigl(
\alpha_{j,-},\alpha_{j,+}
\bigr).
$
\end{proof}

This component-swapping interpretation will be useful in the
description of annihilators and Hermitian duals.

\subsection{Ideal-theoretic consequences}

Recall that every ideal of $
\mathbb Z_{p^e}$
has the form $
p^t\mathbb Z_{p^e},
\,
0\leq t\leq e.$
In particular, every ideal of $\mathbb Z_{p^e}$ is principal.

\begin{corollary}
\label{cor:Rq-PIR}
Under the splitting assumption
$
p_j\equiv1\pmod4\,\,
\text{for every }p_j\mid q,
$
the ring $R_q$ is a finite commutative principal ideal ring.
\end{corollary}

\begin{proof}
By Theorem~\ref{thm:splitting-decomposition},
\[
R_q
\cong
\prod_{j=1}^{s}
\left(
\mathbb Z_{p_j^{e_j}}
\times
\mathbb Z_{p_j^{e_j}}
\right).
\]
Each factor $\mathbb Z_{p_j^{e_j}}$ is a principal ideal ring, and
a finite direct product of principal ideal rings is again a
principal ideal ring.
\end{proof}

\begin{corollary}
\label{cor:Rq-Frobenius}
Under the same hypotheses, $R_q$ is a finite Frobenius ring.
\end{corollary}

\begin{proof}
Every finite commutative principal ideal ring is a Frobenius
ring.
\end{proof}

\section{Principal Ideals, Annihilators, and Hermitian Duality}
\label{sec:duality}

Throughout this section, assume
$
q=\prod_{j=1}^{s}p_j^{e_j},
\,\,
p_j\equiv1\pmod4,
$
and let
$
R_q=\mathbb Z_q[i].
$
By Theorem~\ref{thm:splitting-decomposition},
\[
R_q
\cong
\prod_{j=1}^{s}
\left(
\mathbb Z_{p_j^{e_j}}
\times
\mathbb Z_{p_j^{e_j}}
\right).
\]

For $\alpha\in R_q$, denote the principal ideal generated by
$\alpha$ by
$
\langle\alpha\rangle
=
\{\alpha r:r\in R_q\},
$
and define its annihilator by
$
\Ann(\alpha)
=
\{r\in R_q:r\alpha=0\}.
$

\subsection{Principal ideals and annihilators}

The cardinality of a principal ideal is naturally determined by
the kernel of the corresponding multiplication map.

\begin{proposition}
\label{prop:ideal-ann-cardinality}
For every $\alpha\in R_q$,
\[
|\langle\alpha\rangle|
=
\frac{|R_q|}{|\Ann(\alpha)|}
=
\frac{q^2}{|\Ann(\alpha)|}.
\]
Equivalently,
$
|\langle\alpha\rangle|\,
|\Ann(\alpha)|
=
q^2.
$
\end{proposition}
This identity is valid without invoking the Gaussian norm.

\subsection{CRT description of principal ideals}

Let $\Phi(\alpha)
=
\bigl(
\alpha_{1,+},\alpha_{1,-},
\ldots,
\alpha_{s,+},\alpha_{s,-}
\bigr)$ be the CRT representation of $\alpha$ introduced in
Section~\ref{sec:ring-structure}.

For $x\in\mathbb Z_{p^e}$, define the truncated $p$-adic valuation
by
\[
\nu_p(x)
=
\begin{cases}
t,
&
x\in p^t\mathbb Z_{p^e}
\setminus
p^{t+1}\mathbb Z_{p^e},
\\[1ex]
e,
&
x=0.
\end{cases}
\]
Thus
\[
0\leq \nu_p(x)\leq e.
\]

For each $j$, put
\[
r_j=\nu_{p_j}(\alpha_{j,+}),
\qquad
s_j=\nu_{p_j}(\alpha_{j,-}).
\]

\begin{lemma}
\label{lem:ideal-Zpe}
Let $x\in\mathbb Z_{p^e}$ and let $
t=\nu_p(x).$
Then $
\langle x\rangle
=
p^t\mathbb Z_{p^e},
$
and
$
|\langle x\rangle|
=
p^{e-t}.
$
Moreover,
$
\Ann(x)
=
p^{e-t}\mathbb Z_{p^e},
$
with the conventions
$
p^0\mathbb Z_{p^e}=\mathbb Z_{p^e},
\,
p^e\mathbb Z_{p^e}=\{0\}.
$
Consequently,
$
|\Ann(x)|=p^t.
$
\end{lemma}

\begin{proof}
Write
$
x=p^t u,
$
where $u$ is a unit when $t<e$. Hence
$
\langle x\rangle
=
\langle p^t\rangle
=
p^t\mathbb Z_{p^e},
$
which has $p^{e-t}$ elements.

An element $y\in\mathbb Z_{p^e}$ annihilates $x$ precisely when
$
yp^t u=0\pmod{p^e}.
$
Since $u$ is a unit, this is equivalent to
$
yp^t=0\pmod{p^e},
$
or
$
y\in p^{e-t}\mathbb Z_{p^e}.
$
The case $t=e$, corresponding to $x=0$, follows from the stated
conventions.
\end{proof}

We can now obtain the general cardinality formula.

\begin{theorem}[CRT--valuation cardinality formula]
\label{thm:crt-cardinality}
Let $\alpha\in R_q$ and write
\[
\Phi(\alpha)
=
\bigl(
\alpha_{1,+},\alpha_{1,-},
\ldots,\\
\alpha_{s,+},\alpha_{s,-}
\bigr).
\]
If
$
r_j=\nu_{p_j}(\alpha_{j,+}),
\,
s_j=\nu_{p_j}(\alpha_{j,-}),
$
then
\[
|\langle\alpha\rangle|
=
\prod_{j=1}^{s}
p_j^{\,2e_j-r_j-s_j}.
\]
Furthermore,
\[
|\Ann(\alpha)|
=
\prod_{j=1}^{s}
p_j^{\,r_j+s_j}.
\]
Hence
\[
|\langle\alpha\rangle|\,
|\Ann(\alpha)|
=
q^2.
\]
\end{theorem}
\begin{proof}
Under the CRT isomorphism,
\[
\langle\alpha\rangle
\cong
\prod_{j=1}^{s}
\left(
\langle\alpha_{j,+}\rangle
\times
\langle\alpha_{j,-}\rangle
\right).
\]
By Lemma~\ref{lem:ideal-Zpe},
\[
|\langle\alpha_{j,+}\rangle|
=
p_j^{e_j-r_j},
\qquad
|\langle\alpha_{j,-}\rangle|
=
p_j^{e_j-s_j}.
\]
Therefore
\[
|\langle\alpha\rangle|
=
\prod_{j=1}^{s}
p_j^{e_j-r_j}
p_j^{e_j-s_j}
=
\prod_{j=1}^{s}
p_j^{2e_j-r_j-s_j}.
\]

Similarly,
\[
\Ann(\alpha)
\cong
\prod_{j=1}^{s}
\left(
\Ann(\alpha_{j,+})
\times
\Ann(\alpha_{j,-})
\right),
\]
so
\[
|\Ann(\alpha)|
=
\prod_{j=1}^{s}
p_j^{r_j+s_j}.
\]
Multiplying the two expressions gives
\[
\prod_{j=1}^{s}p_j^{2e_j}
=
q^2.
\]
\end{proof}

The preceding theorem also gives an explicit description of the
annihilator.

\begin{corollary}
\label{cor:ann-crt}
Under the CRT decomposition,
\[
\Ann(\alpha)
\cong
\prod_{j=1}^{s}
\left(
p_j^{\,e_j-r_j}\mathbb Z_{p_j^{e_j}}
\times
p_j^{\,e_j-s_j}\mathbb Z_{p_j^{e_j}}
\right).
\]
\end{corollary}

\subsection{The Gaussian norm and its limitations}
Define a multiplicative map $N:R_q\rightarrow \mathbb{Z}_q$ given by $N(a+ib)=a^2+b^2\,(\text{mod } q)$. Note that $N(a+ib)=0$ does not imply $a+ib=0.$ The map $N$ is referred to as Gaussian norm in \cite{GUZELTEPE2027102881}. 
Since
\[
\alpha_{j,+}=a+b\omega_j,
\qquad
\alpha_{j,-}=a-b\omega_j,
\]
we have
\[
\alpha_{j,+}\alpha_{j,-}
=
a^2-b^2\omega_j^2
=
a^2+b^2
=
N(\alpha)
\pmod{p_j^{e_j}}.
\]

The following result identifies a useful situation in which the
usual norm quotient agrees with the CRT formula.

\begin{proposition}[Validity of the norm formula]
Let $\alpha=a+bi\in R_q$, where $a,b\in\mathbb Z$ are fixed
integral representatives. Suppose that
\[
N(\alpha)=a^2+b^2
   =\prod_{j=1}^{s}p_j^{t_j},
\qquad 0\leq t_j\leq e_j,
\]
and that
\[
t_j=r_j+s_j
\]
for every $j$. Then
\[
|\langle\alpha\rangle|
   =\frac{q^2}{N(\alpha)}.
\]
\end{proposition}

\begin{proof}
Under the stated assumptions,
\[
N(\alpha)
=
\prod_{j=1}^{s}
p_j^{r_j+s_j}.
\]
By Theorem~\ref{thm:crt-cardinality},
\[
|\langle\alpha\rangle|
=
\prod_{j=1}^{s}
p_j^{2e_j-r_j-s_j}.
\]
Therefore
\[
|\langle\alpha\rangle|
=
\frac{
\prod_{j=1}^{s}p_j^{2e_j}
}{
\prod_{j=1}^{s}p_j^{r_j+s_j}
}
=
\frac{q^2}{N(\alpha)}.
\]
\end{proof}

\begin{remark}
\label{rem:norm-warning}
The expression
$
\frac{q^2}{N(\alpha)}
$
is therefore not a general cardinality formula for principal
ideals of $R_q$.  In particular, if
$
v_{p_j}(N(\alpha))>e_j
$
for some $j$, then the Gaussian norm records more $p_j$-adic
divisibility than is visible in the quotient ring
$\mathbb Z_{p_j^{e_j}}$.

\end{remark}

\subsection{Hermitian duality}

For vectors
$
x=(x_1,\ldots,x_n),\,\,
y=(y_1,\ldots,y_n)
\in R_q^n,
$
define the Hermitian inner product by
\[
\langle x,y\rangle_H
=
\sum_{\ell=1}^{n}
x_\ell\sigma(y_\ell).
\]

For an $R_q$-linear code
$
C\subseteq R_q^n,
$
its Hermitian dual is
\[
C^{\perp_H}
=
\left\{
x\in R_q^n:
\langle x,c\rangle_H=0
\text{ for every }c\in C
\right\}.
\]

We first determine the dual of a length-one principal ideal.

\begin{theorem}[Hermitian dual]
\label{thm:principal-hermitian-dual}
Let
$
C=\langle\alpha\rangle\subseteq R_q.
$
Then
$
C^{\perp_H}
=
\Ann(\sigma(\alpha)).
$
\end{theorem}

\begin{proof}
Let $x\in R_q$. Then
$
x\in C^{\perp_H}
$
if and only if
$
\langle x,c\rangle_H=0\,\,
\text{for every }c\in C.
$
Every $c\in C$ has the form
$
c=r\alpha
$
for some $r\in R_q$. Hence
$
\langle x,r\alpha\rangle_H
=
x\,\sigma(r\alpha)
=
x\,\sigma(r)\sigma(\alpha).
$

If
$
x\sigma(\alpha)=0,
$
then the above expression is zero for every $r\in R_q$.
Conversely, if $x\in C^{\perp_H}$, choosing $r=1$ gives
$
x\sigma(\alpha)=0.
$
Therefore
\[
x\in C^{\perp_H}
\iff
x\in\Ann(\sigma(\alpha)),
\]
which proves
\[
C^{\perp_H}
=
\Ann(\sigma(\alpha)).
\]
\end{proof}
\begin{corollary}
\label{cor:dual-cardinality-principal}
If
$
C=\langle\alpha\rangle\subseteq R_q,
$
then
$
|C|\,|C^{\perp_H}|=q^2.
$
In particular,
$
|C^{\perp_H}|
=
\frac{q^2}{|C|}.
$
\end{corollary}

\begin{proof}
By Theorem~\ref{thm:principal-hermitian-dual},
$
C^{\perp_H}
=
\Ann(\sigma(\alpha)).
$
Since $\sigma$ is an automorphism,
$
|\langle\sigma(\alpha)\rangle|
=
|\langle\alpha\rangle|.
$
Applying Proposition~\ref{prop:ideal-ann-cardinality} to
$\sigma(\alpha)$ gives
$
|\langle\sigma(\alpha)\rangle|
\,|\Ann(\sigma(\alpha))|
=
q^2.
$
Thus
\[
|C|\,|C^{\perp_H}|=q^2.
\]
\end{proof}
The CRT decomposition gives an explicit description of the
Hermitian dual.
\begin{corollary}
\label{cor:dual-crt}
Suppose
$
\Phi(\alpha)
=
\bigl(
\alpha_{1,+},\alpha_{1,-},
\ldots,
\alpha_{s,+},\alpha_{s,-}
\bigr)
$
with
\[
r_j=\nu_{p_j}(\alpha_{j,+}),
\qquad
s_j=\nu_{p_j}(\alpha_{j,-}).
\]
Then
$
C^{\perp_H}
=
\Ann(\sigma(\alpha))
$
corresponds under $\Phi$ to
$
\prod_{j=1}^{s}
\left(
p_j^{\,e_j-s_j}\mathbb Z_{p_j^{e_j}}
\times
p_j^{\,e_j-r_j}\mathbb Z_{p_j^{e_j}}
\right).
$
Consequently,
$
|C^{\perp_H}|
=
\prod_{j=1}^{s}p_j^{r_j+s_j}.
$
\end{corollary}

\begin{proof}
By Proposition~\ref{prop:conjugation-crt},
conjugation interchanges the two CRT components:
\[
\Phi(\sigma(\alpha))
=
\bigl(
\alpha_{1,-},\alpha_{1,+},
\ldots,
\alpha_{s,-},\alpha_{s,+}
\bigr).
\]
Applying Corollary~\ref{cor:ann-crt} to $\sigma(\alpha)$ yields the
result.
\end{proof}
\begin{proposition}[Hermitian duality]
\label{prop:hermitian-biduality}
Let $C\subseteq R_q^n$ be an $R_q$-linear code. Then
$
|C|\,|C^{\perp_H}|=|R_q|^n
$
and
$
(C^{\perp_H})^{\perp_H}=C.
$
\end{proposition}

\begin{proof}
Since $R_q$ is a finite Frobenius ring, Euclidean duality satisfies
the standard cardinality and biduality relations \cite{wood1999duality}. Since $\sigma$
is an automorphism of $R_q$, the same relations hold for the
Hermitian dual.
\end{proof}
\section{Non-degenerate Characters}
\label{sec:characters}

To define analogues of Pauli operators to get a CSS-type
construction, we require a suitable additive character of
$
R_q=\mathbb Z_q[i].
$
For $\lambda,\mu\in\mathbb Z_q$ the map  $T_{\lambda,\mu}(a+bi)=\lambda a+\mu b$ is a $\mathbb Z_q$-linear map from
$R_q$ to $\mathbb Z_q$. It induces the additive character $\chi_{\lambda,\mu}:R_q\rightarrow \mathbb{C}^{\times}$
\[
\chi_{\lambda,\mu}(x)
=
\exp\left(
\frac{2\pi i}{q}T_{\lambda,\mu}(x)
\right).
\]
Thus, group of characters of $R_q$ is $\widehat{R_q}=\{\chi_{\lambda, \mu} : \lambda, \mu\in\mathbb{Z}_q\}.$ Moreover, $\widehat{R_q}$ is an $R_q$-module with the following action: $r\in R_q$ and $\chi\in\widehat{R_q}$, $r\cdot\chi(x)=\chi(rx)$.

\begin{proposition}
\label{prop:generating-character}Let $\lambda,\mu\in \mathbb{Z}_q$. The following statements are equivalent:
\begin{enumerate}
 \item The bilinear form 
$
B_{\lambda,\mu}:R_q\times R_q\rightarrow Z_q
$ given by $B_{\lambda,\mu}(x,y)=T_{\lambda,\mu}(xy)$ is non-degenerate.
 \item $\lambda^2+\mu^2$ is a unit in $\mathbb{Z}_q.$
  \item The character $\chi_{\lambda,\mu}$ is a generator of $\widehat{R_q}.$
\end{enumerate}

\end{proposition}

\begin{proof}
Let
$
x=a+bi,
\,\,
y=c+di.
$
Since
$
xy=(ac-bd)+(ad+bc)i,
$
we obtain
\[
T_{\lambda,\mu}(xy)
=
a(\lambda c+\mu d)
+
b(\mu c-\lambda d).
\]
Relative to the basis $\{1,i\}$, the matrix of $B_{\lambda,\mu}$ is
\[
M_{\lambda,\mu}
=
\begin{pmatrix}
\lambda & \mu\\
\mu & -\lambda
\end{pmatrix},
\]
with
\[
\det(M_{\lambda,\mu})
=
-(\lambda^2+\mu^2).
\]
Hence the bilinear form is non-degenerate precisely when
$\lambda^2+\mu^2$ is a unit in $\mathbb Z_q$. Thus $(1)\iff(2).$

For the equivalence of $(2)$ and $(3)$ observe that any $\chi_{\alpha,\beta}\in \widehat{R_q}$ can be written as $\chi_{\alpha,\beta}=(s+it)\chi_{\lambda,\mu}$, for $s=\frac{\lambda a +\mu b}{\lambda^2+\mu^2} , t=\frac{\mu a-\lambda b}{\lambda^2+\mu^2}$.
\end{proof}

Since $q$ is odd,
$
B_{1,1}$
is non-degenerate, since
$
1^2+1^2=2
$ is a unit in $Z_q$.
Henceforth, we fix the generating character
\[
\chi(a+bi)
=
\exp\left(
\frac{2\pi i}{q}(a+b)
\right).
\]
The generating property implies
$
\chi(rx)=1,\text{ for every }r\in R_q
\,\Longrightarrow\,
x=0.
$
This implication will be used in Section~\ref{sec:css} to determine
the kernel of the $X$-error syndrome map.

For $v,x\in R_q^n$, we write
\[
\chi_v(x)
=
\chi\bigl(\langle v,x\rangle_H\bigr).
\]
The non-degeneracy of $\chi$ implies
$
\chi\bigl(\langle v,x\rangle_H\bigr)=1
\,\text{for every }x\in R_q^n
\,\,\Longrightarrow\,\,
v=0.
$
\section{Nested Codes and Quotient Structure}
\label{sec:nested-quotient}
Let
$
C_1\subseteq C_2\subseteq R_q^n
$
be $R_q$-linear codes. A subset
$
\varepsilon\subseteq C_2
$
containing exactly one representative from each coset of $C_1$ in
$C_2$ is called a \emph{transversal of $C_1$ in $C_2$}. Thus
\[
C_2
=
\bigsqcup_{e\in\varepsilon}(C_1+e),
\]
and
\[
|\varepsilon|
=
[C_2:C_1]
=
\frac{|C_2|}{|C_1|}.
\]

The canonical quotient map
\[
\pi_{C_1}:C_2\longrightarrow C_2/C_1,
\qquad
\pi_{C_1}(x)=x+C_1,
\]
has kernel
$
\ker(\pi_{C_1})=C_1.
$
Hence
$
\pi_{C_1}(x)=\pi_{C_1}(y)
\iff
x-y\in C_1.
$
Moreover, if $\varepsilon$ is a transversal, then
$\pi_{C_1}|_{\varepsilon}$ gives a bijection between
$\varepsilon$ and $C_2/C_1$.






\subsection{Nested principal ideals}

For principal ideal codes, the quotient index can be
expressed directly in terms of annihilator cardinalities.

\begin{proposition}
\label{prop:index-annihilator}
Let
$
C_1=\langle\alpha_1\rangle
\subseteq
C_2=\langle\alpha_2\rangle
\subseteq R_q.
$
Then
$
[C_2:C_1]
=
\frac{|\Ann(\alpha_1)|}
     {|\Ann(\alpha_2)|}.
$
\end{proposition}

\begin{proof}
By Proposition~\ref{prop:ideal-ann-cardinality},
$
|C_i|
=
\frac{q^2}{|\Ann(\alpha_i)|},
\qquad i=1,2.
$
Therefore
\[
[C_2:C_1]
=
\frac{|C_2|}{|C_1|}
=
\frac{
q^2/|\Ann(\alpha_2)|
}{
q^2/|\Ann(\alpha_1)|
}
=
\frac{|\Ann(\alpha_1)|}
     {|\Ann(\alpha_2)|}.
\]
\end{proof}

Combining this with the CRT--valuation formula gives an explicit
index expression.

\begin{corollary}
\label{cor:index-valuations}
Let
$
C_1=\langle\alpha_1\rangle
\subseteq
C_2=\langle\alpha_2\rangle,
$
and write
$
\Phi(\alpha_\ell)
=
\bigl(
\alpha_{\ell,1,+},
\alpha_{\ell,1,-},
\ldots,
\alpha_{\ell,s,+},\\
\alpha_{\ell,s,-}
\bigr),
\ell=1,2.
$
Set
$
r_{\ell,j}
=
\nu_{p_j}(\alpha_{\ell,j,+}),
\,
s_{\ell,j}
=
\nu_{p_j}(\alpha_{\ell,j,-}).
$
Then
\[
[C_2:C_1]
=
\prod_{j=1}^{s}
p_j^{
(r_{1,j}+s_{1,j})
-
(r_{2,j}+s_{2,j})
}.
\]
\end{corollary}

\begin{proof}
By Theorem~\ref{thm:crt-cardinality},
$
|\Ann(\alpha_\ell)|
=
\prod_{j=1}^{s}
p_j^{r_{\ell,j}+s_{\ell,j}},
\,
\ell=1,2.
$
Substitution into Proposition~\ref{prop:index-annihilator} gives the
result.
\end{proof}
The stabilizer-theoretic interpretation of $C_2/C_1$ will be
established in Section~6.


\section{CSS-Type Stabilizer Codes and Syndrome Structure}
\label{sec:css}

Throughout this section, let
$
C_1\subseteq C_2\subseteq R_q^n
$
be $R_q$-linear codes satisfying
$
C_2^{\perp_H}\subseteq C_1\subseteq C_2.
$

Fix a generating character
\[
\chi:R_q\longrightarrow\mathbb C^\times
\]
as constructed in Section~\ref{sec:characters}.

Let $\mathcal H
=
\mathbb C^{|R_q|^n}$
with computational basis
$
\{
|x\rangle:x\in R_q^n
\}.
$
\subsection{Generalized Pauli operators}

For $u,v\in R_q^n,$
define operators $X(u),Z(v):\mathcal H\longrightarrow\mathcal H
$
by
\[
X(u)|x\rangle
=
|x+u\rangle
\qquad\text{and}\qquad
Z(v)|x\rangle
=
\chi\bigl(\langle v,x\rangle_H\bigr)|x\rangle.
\]
\begin{remark}
 For every $u,v\in R_q^n$,
$
Z(v)X(u)
=
\chi\bigl(\langle v,u\rangle_H\bigr)
X(u)Z(v).
$    
\end{remark}
\begin{remark}
    If
$
Z(v)=Z(v'),
$
then
$
v=v'.
$
Moreover,
$
X(u)=X(u')
\Longrightarrow
u=u'.
$ Hence the phase operators are faithfully
labelled by $R_q^n$.
\end{remark}




\subsection{CSS-type stabilizer structure}

Define
\[
\mathcal S
=
\left\langle
X(u),Z(v):
u\in C_1,\;
v\in C_2^{\perp_H}
\right\rangle .
\]
\begin{theorem}[CSS-type stabilizer construction]
\label{thm:CSS-construction}
Suppose
$
C_2^{\perp_H}
\subseteq
C_1
\subseteq
C_2
\subseteq
R_q^n.
$
Then:

\begin{enumerate}

\item[(i)]
$\mathcal S$ is an abelian subgroup of the generalized Pauli
group.

\item[(ii)]
Its cardinality is
$
|\mathcal S|
=
|C_1|\,
|C_2^{\perp_H}|.
$
\item[(iii)]
The simultaneous $+1$ eigenspace of $\mathcal S$ has complex
dimension
$
K
=
\frac{|C_2|}{|C_1|}.
$
\item[(iv)] The quotient of the $X$-type normalizer subgroup by the
$X$-type stabilizer subgroup is naturally isomorphic to
$
\frac{C_2}{C_1}.
$
\item[(v)]
The group of logical $Z$-operators modulo the $Z$-type
stabilizer subgroup is naturally isomorphic to
$
C_1^{\perp_H}/C_2^{\perp_H}.
$
\end{enumerate}
\end{theorem}
\begin{proof}
Let
$
u\in C_1
\,\,\text{and}\,\,
v\in C_2^{\perp_H}.
$
Since
$
C_1\subseteq C_2,
$
we have
$
\langle v,u\rangle_H=0.
$
Therefore $
Z(v)X(u)=X(u)Z(v),
$
so $\mathcal S$ is abelian.

Since the generators commute,
every element of $\mathcal S$ can be written in the form
\[
X(u)Z(v),
\qquad
u\in C_1,\quad
v\in C_2^{\perp_H}.
\]
This representation is unique. Indeed, if
\[
X(u)Z(v)=X(u')Z(v')\quad\implies X(u-u')Z(v-v')=I.
\]

Acting on a computational basis vector shows that
$
u-u'=0,
$
since a nonzero translation cannot fix every basis vector.
Hence
$
v-v'=0.
$
Therefore
$
|\mathcal S|
=
|C_1|\,|C_2^{\perp_H}|.
$

Let
\[
P_{\mathcal S}
=
\frac{1}{|\mathcal S|}
\sum_{S\in\mathcal S}S.
\]
This is the projector onto the simultaneous $+1$ eigenspace of
$\mathcal S$.

Every nonidentity element $X(u)Z(v)$ of $\mathcal S$ has trace
zero. If $u\neq0$, the translation $x\mapsto x+u$ has no fixed
computational basis vector, so
\[
\operatorname{Tr}(X(u)Z(v))=0.
\]
If $u=0$ and $v\neq0$, then
\[
\operatorname{Tr}(Z(v))
=
\sum_{x\in R_q^n}
\chi\bigl(\langle v,x\rangle_H\bigr)
=
0,
\]
since $x\mapsto\chi(\langle v,x\rangle_H)$ is a nontrivial
character of the additive group $R_q^n$. Thus only the identity
has nonzero trace, with
\[
\operatorname{Tr}(I)=|R_q|^n.
\]
Therefore
\[
K
=
\operatorname{Tr}(P_{\mathcal S})
=
\frac{|R_q|^n}{|\mathcal S|}.
\]
By Proposition~\ref{prop:hermitian-biduality},
\[
|C_2|\,|C_2^{\perp_H}|
=
|R_q|^n.
\]
Thus
\[
K
=
\frac{|C_2|\,|C_2^{\perp_H}|}
{|C_1|\,|C_2^{\perp_H}|}
=
\frac{|C_2|}{|C_1|}.
\]

We next determine the logical $X$-operators. An operator $X(a)$
commutes with every $Z(v)$,
\[
v\in C_2^{\perp_H},
\]
if and only if
\[
\chi\bigl(\langle v,a\rangle_H\bigr)=1
\]
for every $v\in C_2^{\perp_H}$.

Since $C_2^{\perp_H}$ is $R_q$-linear and $\chi$ is generating,
this is equivalent to
\[
\langle v,a\rangle_H=0
,\,\,
\text{for every }v\in C_2^{\perp_H}.
\]
Hence
$
a\in
(C_2^{\perp_H})^{\perp_H}.
$
By Proposition~\ref{prop:hermitian-biduality},
$
(C_2^{\perp_H})^{\perp_H}=C_2.
$
Thus the $X$-type operators in the normalizer are precisely
$
X(C_2).
$
The $X$-type stabilizers are
$
X(C_1).
$

Therefore
\[
\frac{\{\text{$X$-type normalizer operators}\}}
{\{\text{$X$-type stabilizer operators}\}}
\cong
C_2/C_1.
\]

Similarly, $Z(b)$ commutes with every $X(u)$,
$u\in C_1$, precisely when
$
b\in C_1^{\perp_H}.
$
The $Z$-type stabilizers correspond to
$
C_2^{\perp_H}.
$
Therefore the logical $Z$-operator quotient is
$
C_1^{\perp_H}/C_2^{\perp_H}.
$
\end{proof}


\subsection{\texorpdfstring{$X$}{X}-error stabilizer syndrome}

Let $e\in R_q^n.$
For the $X$-type error $X(e)$, define its stabilizer syndrome by
its commutation phases with the $Z$-type stabilizers.

\begin{definition}
The $X$-error syndrome map is
\[
\operatorname{Syn}_X:
R_q^n
\longrightarrow
\operatorname{Hom}
\bigl(
C_2^{\perp_H},
\mathbb C^\times
\bigr),
\]
defined by
\[
\operatorname{Syn}_X(e)(v)
=
\chi\bigl(\langle v,e\rangle_H\bigr),
\qquad
v\in C_2^{\perp_H}.
\]
\end{definition}

The kernel of this map determines precisely which $X$-type errors
have the same stabilizer syndrome.
\begin{theorem}[Syndrome kernel for $X$-type errors]
\label{thm:X-syndrome-kernel}
For the CSS-type stabilizer of
Theorem~\ref{thm:CSS-construction},
\[
\ker(\operatorname{Syn}_X)
=
C_2.
\]
Consequently,
$
\operatorname{Syn}_X(e_1)
=
\operatorname{Syn}_X(e_2)
\iff
e_1-e_2\in C_2.
$
\end{theorem}

\begin{proof}
Let $e\in\ker(\operatorname{Syn}_X).$
Then
$
\chi\bigl(\langle v,e\rangle_H\bigr)=1,
$ for every $v\in C_2^{\perp_H}.$ Since $C_2^{\perp_H}$ is $R_q$-linear, for every $r\in R_q$ and every $
v\in C_2^{\perp_H},
$ we also have $
rv\in C_2^{\perp_H}.
$
Therefore
\[
1
=
\chi\bigl(\langle rv,e\rangle_H\bigr).
\]
Using $R_q$-linearity in the first component,
\[
\langle rv,e\rangle_H
=
r\langle v,e\rangle_H.
\]
Hence
\[
\chi\bigl(r\langle v,e\rangle_H\bigr)=1,\,\text{for every } r\in R_q.\]
Since $\chi$ is generating,
\[
\chi(ra)=1,
\,\,
\text{for every }r\in R_q
\implies a=0.\]
Thus
\[
\langle v,e\rangle_H=0, \text{ for every } v\in C_2^{\perp_H}.
\]
Consequently,
\[
e\in
(C_2^{\perp_H})^{\perp_H}.
\]
By Proposition~\ref{prop:hermitian-biduality},
\[
(C_2^{\perp_H})^{\perp_H}=C_2.
\]
Hence
\[
\ker(\operatorname{Syn}_X)\subseteq C_2.
\]

Conversely, let
$
e\in C_2.
$
Then
$
\langle v,e\rangle_H=0,
$
for every
$
v\in C_2^{\perp_H}.
$
Therefore
$
\operatorname{Syn}_X(e)(v)
=
\chi(0)
=
1
$
for every $v$, and hence
$
e\in\ker(\operatorname{Syn}_X).
$
Thus
\[
C_2\subseteq\ker(\operatorname{Syn}_X).
\]

Combining the two inclusions gives
$
\ker(\operatorname{Syn}_X)=C_2.
$

Finally,
$
\operatorname{Syn}_X(e_1)
=
\operatorname{Syn}_X(e_2)
$
if and only if
$
\operatorname{Syn}_X(e_1-e_2)
=
1,
$
which is equivalent to
$
e_1-e_2\in C_2.
$
\end{proof}


\begin{corollary}
\label{cor:X-syndrome-quotient}
The image of the $X$-error syndrome map satisfies
$
\operatorname{Im}(\operatorname{Syn}_X)
\cong R_q^n/C_2.
$
Consequently,
$
|\operatorname{Im}(\operatorname{Syn}_X)|
=
\frac{|R_q|^n}{|C_2|}
=
|C_2^{\perp_H}|.
$
\end{corollary}

\begin{proof}
By Theorem~\ref{thm:X-syndrome-kernel} and the first isomorphism
theorem,
$
\operatorname{Im}(\operatorname{Syn}_X)
\cong R_q^n/C_2.
$
Hence
$
|\operatorname{Im}(\operatorname{Syn}_X)|
=
\frac{|R_q|^n}{|C_2|}
=
|C_2^{\perp_H}|,$
where the last equality follows from Hermitian duality.
\end{proof}
\begin{corollary}[Three-level $X$-operator structure]
\label{cor:three-level-structure}
Let
$
C_2^{\perp_H}\subseteq C_1\subseteq C_2\subseteq R_q^n
$
and let $\mathcal S$ be the associated CSS-type stabilizer. Then
the chain
$
C_1\subseteq C_2\subseteq R_q^n
$
has the following stabilizer-theoretic interpretation:
\[
\begin{array}{rcl}
C_1
&\longleftrightarrow&
\text{$X$-type stabilizer labels},\\[1mm]
C_2/C_1
&\longleftrightarrow&
\text{$X$-type normalizer operators modulo $X$-type stabilizers},\\[1mm]
R_q^n/C_2
&\longleftrightarrow&
\text{physical $X$-error syndrome classes}.
\end{array}
\]
In particular, $C_2/C_1$ and $R_q^n/C_2$ describe different
equivalence relations.
\end{corollary}

\begin{proof}
The first two assertions follow from the $X$-type stabilizer and
normalizer structure, while the third follows from
Theorem~\ref{thm:X-syndrome-kernel}, since
$
\ker(\operatorname{Syn}_X)=C_2.
$
\end{proof}
\subsection{Correction of coset--syndrome distinguishability}
\label{subsec:coset-syndrome-correction}

We now apply the distinction established in
Corollary~\ref{cor:three-level-structure} to a transversal
\[
C_2=\bigsqcup_{e\in\varepsilon}(C_1+e),
\qquad \varepsilon\subseteq C_2,
\]
of $C_1$ in $C_2$.

\begin{proposition}
\label{prop:corrected-coset-syndrome}
Let
\[
C_2^{\perp_H}\subseteq C_1\subseteq C_2\subseteq R_q^n,
\]
and let $\varepsilon$ be a transversal of $C_1$ in $C_2$. Then:

\begin{enumerate}
\item distinct elements of $\varepsilon$ represent distinct
classes in $C_2/C_1$;

\item every element of $\varepsilon$ has the same physical
$X$-error stabilizer syndrome.
\end{enumerate}

Consequently, $\varepsilon$ parametrizes distinct logical
$X$-operator classes, but not distinct physical $X$-error syndrome
classes.
\end{proposition}

\begin{proof}
Since $\varepsilon$ is a transversal, if
$
e_i,e_j\in\varepsilon,
\,\,
e_i\neq e_j,
$
then
$
e_i-e_j\notin C_1.
$
Hence $e_i$ and $e_j$ determine distinct elements of $C_2/C_1$.

On the other hand,
$
e_i,e_j\in C_2,
$
 so
$
e_i-e_j\in C_2.
$
By Theorem~\ref{thm:X-syndrome-kernel},
$
\operatorname{Syn}_X(e_i)
=
\operatorname{Syn}_X(e_j).
$
Thus all elements of $\varepsilon$ have the same physical
$X$-error syndrome.
\end{proof}
\subsection{General Pauli errors and the stabilizer correctability criterion}
A generalized Pauli error has the form
$
E(a,b)=X(a)Z(b),
\text{ for }a,b\in R_q^n.
$
Let $\mathcal P$ denote the generalized Pauli group and let
$N(\mathcal S)$ denote the normalizer of $\mathcal S$ in
$\mathcal P$.

Since global scalar phases have no physical effect, generalized
Pauli operators will henceforth be identified modulo global phases.

A generalized Pauli error has trivial stabilizer syndrome precisely
when it commutes with every element of $\mathcal S$. Thus the kernel
of the full stabilizer syndrome map is the Pauli centralizer of
$\mathcal S$, which, under the above convention, coincides with its
Pauli normalizer $N(\mathcal S)$.

We now recall the standard stabilizer correctability criterion
\cite{gottesman1997stabilizer}.
\begin{theorem}[Correctability criterion \cite{gottesman1997stabilizer}]
\label{thm:Pauli-correctability}
Let
$
\mathcal E=\{E_a\}
$
be a set of generalized Pauli errors. Then $\mathcal E$ is
correctable by the stabilizer code determined by $\mathcal S$ if,
for every pair
$
E_a,E_b\in\mathcal E,
$
one has
\[
E_a^\dagger E_b
\notin
N(\mathcal S)\setminus\mathcal S.
\]
Equivalently, for every pair of errors, either
\[
E_a^\dagger E_b\in\mathcal S,
\]
in which case the two errors act identically on the code space, or
\[
E_a^\dagger E_b\notin N(\mathcal S),
\]
in which case they produce distinct stabilizer syndromes.
\end{theorem}
\begin{proposition}[Correctability criterion for $X$-type errors]
\label{cor:X-error-correctability}
Let
$
E\subseteq R_q^n
$
and consider the family of $X$-type errors
$
\mathcal{E}_X=\{X(e):e\in E\}.
$ Then $\mathcal{E}_X$ is correctable by the stabilizer code determined by
$\mathcal{S}$ if and only if, for every $e_i,e_j\in E$,
\[
e_j-e_i\in C_1
\qquad\text{or}\qquad
e_j-e_i\notin C_2.
\]
Equivalently,
$
e_j-e_i\notin C_2\setminus C_1
\,
\text{ for all }e_i,e_j\in E.
$
\end{proposition}

\begin{proof}
For two $X$-type errors $X(e_i)$ and $X(e_j)$, we have
\[
X(e_i)^\dagger X(e_j)
=
X(-e_i)X(e_j)
=
X(e_j-e_i).
\]
By Theorem~\ref{thm:CSS-construction}, the $X$-type stabilizers are precisely
the operators $X(c)$ with $c\in C_1$, whereas the $X$-type operators in
the normalizer are precisely the operators $X(c)$ with $c\in C_2$.

Therefore,
\[
X(e_j-e_i)\in\mathcal{S}
\quad\Longleftrightarrow\quad
e_j-e_i\in C_1,
\]
while
\[
X(e_j-e_i)\in N(\mathcal{S})\setminus\mathcal{S}
\quad\Longleftrightarrow\quad
e_j-e_i\in C_2\setminus C_1.
\]
By Theorem~\ref{thm:Pauli-correctability}, the error family
$\mathcal{E}_X$ is correctable precisely when
\[
X(e_i)^\dagger X(e_j)
\notin
N(\mathcal{S})\setminus\mathcal{S}
\]
for every pair $e_i,e_j\in E$. Hence the forbidden case is exactly
\[
e_j-e_i\in C_2\setminus C_1.
\]
Equivalently, for every pair $e_i,e_j\in E$, either
\[
e_j-e_i\in C_1\qquad
\text{or}\qquad
e_j-e_i\notin C_2.
\]
\end{proof}
\begin{remark}
Two errors having the same stabilizer syndrome need not differ by
a stabilizer element.  They may instead differ by a nontrivial
logical operator in $
N(\mathcal S)\setminus\mathcal S.$
Therefore equality of syndrome alone is not sufficient to conclude
that two physical errors are equivalent on the code space.
\subsection{Syndrome representatives and an \texorpdfstring{$X$}{X}-error recovery construction}
\label{subsec:syndrome-recovery}

The preceding results identify $R_q^n/C_2$, rather than $C_2/C_1$,
as the quotient governing physical $X$-error syndromes. This suggests
a natural choice of error representatives based directly on the syndrome
equivalence relation.
\end{remark}
\begin{definition}[Syndrome transversal]
\label{def:syndrome-transversal}
A subset $\Gamma\subseteq R_q^n$ is called an $X$-syndrome transversal
if it contains exactly one representative from each coset of $C_2$ in
$R_q^n$. Thus
\[
R_q^n
=
\bigsqcup_{\gamma\in\Gamma}
(\gamma+C_2).
\]
Equivalently,
\[
|\Gamma|
=
[R_q^n:C_2]
=
\frac{|R_q|^n}{|C_2|}
=
|C_2^{\perp_H}|.
\]
\end{definition}

The following proposition shows that, unlike a transversal of $C_1$
inside $C_2$, a syndrome transversal genuinely consists of
syndrome-distinguishable $X$-errors.

\begin{proposition}
\label{prop:syndrome-transversal}
Let $\Gamma$ be an $X$-syndrome transversal. Then for any distinct
$\gamma_1,\gamma_2\in\Gamma$,
\[
\operatorname{Syn}_X(\gamma_1)
\neq
\operatorname{Syn}_X(\gamma_2).
\]
Consequently, the restriction
\[
\operatorname{Syn}_X|_{\Gamma}:
\Gamma\longrightarrow
\operatorname{Im}(\operatorname{Syn}_X)
\]
is a bijection.
\end{proposition}

\begin{proof}
Let $\gamma_1,\gamma_2\in\Gamma$ be distinct. Since $\Gamma$ contains
exactly one representative from each coset of $C_2$ in $R_q^n$, we have
\[
\gamma_1+C_2\neq \gamma_2+C_2.
\]
Hence
\[
\gamma_1-\gamma_2\notin C_2.
\]
By Theorem~\ref{thm:X-syndrome-kernel},
\[
\operatorname{Syn}_X(\gamma_1)
=
\operatorname{Syn}_X(\gamma_2)
\quad\Longleftrightarrow\quad
\gamma_1-\gamma_2\in C_2.
\]
Therefore
\[
\operatorname{Syn}_X(\gamma_1)
\neq
\operatorname{Syn}_X(\gamma_2).
\]
Thus $\operatorname{Syn}_X|_{\Gamma}$ is injective. Since
\[
|\Gamma|
=
[R_q^n:C_2]
=
|\operatorname{Im}(\operatorname{Syn}_X)|,
\]
it is bijective.
\end{proof}

We now obtain a correctable error family directly from the physical
syndrome quotient.

\begin{theorem}[Syndrome-transversal \texorpdfstring{$X$}{X}-error construction]
\label{thm:syndrome-transversal-correctable}
Let $\Gamma$ be an $X$-syndrome \\transversal of $C_2$ in $R_q^n$, and
define
\[
\mathcal{E}_{\Gamma}
=
\{X(\gamma):\gamma\in\Gamma\}.
\]
Then $\mathcal{E}_{\Gamma}$ is a correctable family of $X$-type errors
for the stabilizer code determined by $\mathcal{S}$.
\end{theorem}

\begin{proof}
Let $\gamma_1,\gamma_2\in\Gamma$.

If $\gamma_1=\gamma_2$, then
\[
\gamma_2-\gamma_1=0\in C_1.
\]

Suppose now that $\gamma_1\neq\gamma_2$. Since they represent distinct
cosets of $C_2$,
\[
\gamma_2-\gamma_1\notin C_2.
\]
Thus, for every pair $\gamma_1,\gamma_2\in\Gamma$, either
\[
\gamma_2-\gamma_1\in C_1
\qquad\text{or}\qquad
\gamma_2-\gamma_1\notin C_2.
\]
By Proposition~\ref{cor:X-error-correctability},
the error family $\mathcal{E}_{\Gamma}$ is correctable.
\end{proof}

\begin{remark}
\label{rem:scope-syndrome-recovery}
Theorem~\ref{thm:syndrome-transversal-correctable} does not assert
that all $X$-type errors in $R_q^n$ are correctable. Rather, it
identifies the prescribed family
$
\mathcal{E}_{\Gamma}=\{X(\gamma):\gamma\in\Gamma\}
$
as a correctable family. Indeed, if an arbitrary error has the form
$
e=\gamma+c,\text{ where }\gamma\in\Gamma\text{ and } c\in C_2,
$
then $e$ and $\gamma$ have the same stabilizer syndrome. Applying
$X(-\gamma)$ leaves the residual operator
$
X(-\gamma)X(e)=X(c).
$
This residual operator acts trivially on the code space if and only if
$c\in C_1$, whereas $c\in C_2\setminus C_1$ gives a nontrivial
logical $X$-operator. Thus syndrome measurement alone cannot, in
general, correct all errors in a syndrome coset $\gamma+C_2$.
\end{remark}
For the prescribed error family $\mathcal{E}_{\Gamma}$, the
syndrome-based recovery is well defined. Indeed, by
Proposition~\ref{prop:syndrome-transversal}, every measured syndrome
corresponds to a unique $\gamma\in\Gamma$. The recovery operator
associated with the syndrome $s=\operatorname{Syn}_X(\gamma)$ is
\[
\mathcal{R}_s=X(-\gamma),
\]
and hence
\[
\mathcal{R}_sX(\gamma)=I.
\]
Thus every error in $\mathcal{E}_{\Gamma}$ is recovered exactly.

The recovery procedure may be summarized as follows:
\begin{enumerate}
    \item Measure the stabilizer syndrome
    \[
    s=\operatorname{Syn}_X(e).
    \]

    \item Determine the unique $\gamma\in\Gamma$ satisfying
    \[
    \operatorname{Syn}_X(\gamma)=s.
    \]

    \item Apply the recovery operator
    \[
    X(-\gamma).
    \]
\end{enumerate}
The syndrome transversal $\Gamma$ is not unique; in applications,
its representatives may be chosen according to an additional
criterion, such as minimum weight or maximum likelihood for a
specified noise model.
\section{Examples}
\label{sec:examples}

We illustrate the preceding results through explicit examples over rings $\mathbb{Z}_q[i]$. We present a new example, namely, over $\mathbb{Z}_{125}$ and revisit an example in \cite{GUZELTEPE2027102881}, which in particular demonstrates the failure of the formula they are using to compute cardinality of a principal ideal of the ring $R_q.$ Further, these examples also demonstrate the distinction
between logical and physical syndrome quotients.

\subsection{Revisiting the example of \cite{GUZELTEPE2027102881} over \texorpdfstring{$\mathbb Z_{85}[i]$}{Z_{85}[i]}}
\label{subsec:example-Z85}

Consider
$
R_{85}=\mathbb Z_{85}[i],
\,
C_2=\langle2+i\rangle,
\,
C_1=\langle7+6i\rangle.
$
Since
$
(2+i)(4+i)=7+6i,
$
we have
$
C_1\subseteq C_2.
$

Using
$
N(2+i)=5,
\,
N(7+6i)=85,
$
we obtain
$
|C_2|=\frac{85^2}{5}=1445,
\,
|C_1|=\frac{85^2}{85}=85.
$
Consequently,
\[
|C_2/C_1|=17.
\]

Furthermore,
$
C_2^{\perp_H}
=
\Ann(2-i).
$
Since
$
(34+17i)(2-i)=85\equiv0\pmod{85}
$
and
$
|C_2^{\perp_H}|
=
\frac{7225}{1445}
=
5,
$
we obtain
\[
C_2^{\perp_H}
=
\langle34+17i\rangle.
\]
Also,
\[
17(7+6i)
\equiv
34+17i
\pmod{85},
\]
so
\[
C_2^{\perp_H}
\subseteq
C_1
\subseteq
C_2.
\]

The number of logical $X$-operator classes is therefore
\[
|C_2/C_1|=17,
\]
whereas the total number of physical $X$-error syndrome classes is
\[
|R_{85}/C_2|
=
\frac{7225}{1445}
=
5.
\]
Equivalently,
\[
|\operatorname{Im}(\operatorname{Syn}_X)|
=
|C_2^{\perp_H}|
=
5.
\]

Thus
\[
|C_2/C_1|=17,
\qquad\text{whereas}\qquad
|R_{85}/C_2|=5.
\]
Hence $C_2/C_1$ and $R_{85}/C_2$ describe two distinct objects in the
stabilizer construction. The quotient $C_2/C_1$ parametrizes the
logical $X$-operator classes, whereas $R_{85}/C_2$ parametrizes the
physical $X$-error syndrome classes. In particular, these two quotient
spaces cannot be identified in this example.

The distinction is also seen directly from
$
C_2=\ker(\operatorname{Syn}_X).
$
Indeed, every $e\in C_2$, and hence every representative of a
transversal of $C_1$ in $C_2$, has trivial $X$-error syndrome, i.e.
$
\operatorname{Syn}_X(e)=1.
$
Thus distinct representatives of $C_2/C_1$ correspond to distinct
logical $X$-operator classes, but they are not distinguished by
stabilizer syndrome measurements.

\subsection{Failure of the norm shortcut given in \cite{GUZELTEPE2027102881} over
$\mathbb Z_{125}[i]$}
\label{subsec:example-Z125}

Let
$
R_{125}=\mathbb Z_{125}[i].
$
Since
$
57^2\equiv-1\pmod{125},
$
we use the splitting map
$$
\Phi(a+bi)
=
(a+57b,\;a-57b).
$$

Consider
$
C_2=\langle2+i\rangle,
\,
C_1=\langle25+25i\rangle.
$
Since
$
(2+i)(15+5i)=25+25i,
$
we have
$
C_1\subseteq C_2.
$

For $2+i$,
$
\Phi(2+i)=(59,70),
$
so that
$
\nu_5(59)=0,
\,
\nu_5(70)=1.
$
Theorem~\ref{thm:crt-cardinality} therefore gives
\[
|C_2|
=
5^{6-0-1}
=
3125.
\]

For
$
\alpha_1=25+25i, $
the Gaussian norm is
$
N(\alpha_1)
=
25^2+25^2
=
1250.$ Thus an unrestricted application of the norm quotient would give
\[
\frac{125^2}{1250}
=
\frac{25}{2},
\]
which is not an integer and therefore cannot be the cardinality of an ideal.

On the other hand,
\[
\Phi(25+25i)
=
(75,100),
\]
with
\[
\nu_5(75)=2,
\qquad
\nu_5(100)=2.
\]
Hence the CRT--valuation formula gives
\[
|C_1|
=
5^{6-2-2}
=
25.
\]
Consequently,
\[
[C_2:C_1]
=
\frac{3125}{25}
=
125.
\]

For completeness,
\[
C_2^{\perp_H}
=
\Ann(2-i).
\]
The element $
\beta=75+100i $
satisfies $\beta(2-i)
=
250+125i
\equiv0
\pmod{125}.$ 
Moreover,
\[
\Phi(\beta)=(25,0),
\]
so
\[
|\langle\beta\rangle|
=
5^{6-2-3}
=
5.
\]
Since
\[
|C_2^{\perp_H}|
=
\frac{15625}{3125}
=
5,
\]
we have
\[
C_2^{\perp_H}
=
\langle75+100i\rangle.
\]
Finally,
\[
(1+3i)(25+25i)
\equiv
75+100i
\pmod{125},
\]
and hence
\[
C_2^{\perp_H}
\subseteq
C_1
\subseteq
C_2.
\]

The principal purpose of this example is therefore the failure of
the unrestricted norm expression:
\[
\frac{q^2}{N(\alpha)}
\text{ need not give }
|\langle\alpha\rangle|,
\]
whereas the CRT--valuation formula remains valid.

We now apply the syndrome-kernel description of Section~6 to the
coset construction over
\[
R_{325}=\mathbb{Z}_{325}[i].
\]
Consider
\[
C_2=\langle \alpha_2\rangle,
\qquad
\alpha_2=3+2i,
\]
and
\[
C_1=\langle \alpha_1\rangle,
\qquad
\alpha_1=4+7i.
\]
Since
\[
\alpha_1=(2+i)\alpha_2,
\]
we have $C_1\subseteq C_2$. Moreover,
\[
|C_2|=8125,
\qquad
|C_1|=1625, \textnormal{ and hence }
|C_2/C_1|=5.
\]
A transversal of $C_1$ in $C_2$ is
$\varepsilon
=
\{0,\alpha_2,2\alpha_2,3\alpha_2,4\alpha_2\}.
$

Algebraically, the five elements of $\varepsilon$ represent the five
distinct classes of $C_2/C_1$. However, since
\[
\varepsilon\subseteq C_2=\ker(\operatorname{Syn}_X),
\]
we have
\[
\operatorname{Syn}_X(e)=\operatorname{Syn}_X(0)
\qquad\text{for every }e\in\varepsilon.
\]
Thus the five elements of $\varepsilon$ represent distinct logical
$X$-operator classes modulo the $X$-type stabilizer subgroup, but they
do not represent distinct physical $X$-error syndromes. Indeed,
\[
|C_2/C_1|=5,
\qquad
|R_{325}/C_2|
=\frac{325^2}{8125}
=13.
\]
\section{Conclusion}

We studied a ring-theoretic and stabilizer-theoretic framework for
CSS-type constructions over the ring
$R_q=\mathbb{Z}_q[i]$ under the splitting assumption on the prime
divisors of $q$. The CRT decomposition yields a valuation-based
cardinality formula for principal ideals and shows that the Hermitian
dual of a length-one principal ideal code $C=\langle\alpha\rangle$ is
$
C^{\perp_H}=\operatorname{Ann}(\sigma(\alpha)).
$

For nested codes
$
C_2^{\perp_H}\subseteq C_1\subseteq C_2,
$
we showed that the physical $X$-error syndrome map has kernel $C_2$.
Consequently, $C_2/C_1$ describes logical $X$-operator classes,
whereas $R_q^n/C_2$ describes physical $X$-error syndrome classes.
This distinction explains why a transversal of $C_1$ in $C_2$
does not, in general, provide syndrome-distinguishable error
representatives.

Using the physical syndrome quotient, we introduced an
$X$-syndrome transversal $\Gamma$ of $C_2$ in $R_q^n$ and showed
that
$
\mathcal{E}_{\Gamma}
=
\{X(\gamma):\gamma\in\Gamma\}
$
forms a correctable family of $X$-type errors. The resulting recovery
procedure identifies the unique representative associated with the
measured syndrome and applies the corresponding inverse translation.
The examples illustrate these structural distinctions and the role of
the CRT--valuation description in explicit computations.
\section*{Conflict of Interest} Both the authors declare that they have no conflict of interest.
\section*{Acknowledgments}
The first author would like to acknowledge PMRF (PMRF Id: 1403187) for its financial support.
\bibliographystyle{abbrv}
\bibliography{references}






\end{document}